\documentclass[12pt]{amsart}
\usepackage{amsmath,amssymb,amsfonts,amsthm}
\usepackage{graphicx}
\usepackage{tikz}
\usepackage{tabularx}
\usepackage[colorlinks]{hyperref}

\newtheorem{theorem}{Theorem}

\newtheorem{lemma}[theorem]{Lemma}
\newtheorem{proposition}[theorem]{Proposition}

\theoremstyle{definition}
\theoremstyle{remark}\newtheorem{remark}[theorem]{Remark}

\def\Ito/{It\^o}

\def\aa{\mathbf{a}}
\def\bb{\mathbf{b}}

\newcommand{\arxiv}[1]{{\tt \href{http://arxiv.org/abs/#1}{arXiv:#1}}}

\title[A duality principle for selection games]{A duality principle for selection games}
\author[Levine, Sheffield, Stange]{Lionel Levine, Scott Sheffield and Katherine E. Stange}

\date{February 13, 2012}
\address{Department of Mathematics, Cornell University, Ithaca, NY 14853. \url{http://www.math.cornell.edu/~levine}}
\email{levine@math.cornell.edu}
\address{MIT Department of Mathematics, 77 Massachusetts Ave., Cambridge, MA 02139. \url{http://math.mit.edu/~sheffield}}
\email{sheffield@math.mit.edu}
\address{Department of Mathematics, Stanford University, 450 Serra Mall, Building 380, Stanford, CA 94305 USA. \url{http://math.stanford.edu/~stange}}
\email{stange@math.stanford.edu}

\thanks{
The first author was supported by NSF MSPRF 0803064.
The second author was partially supported by NSF grant DMS 0645585.
The third author was supported by NSF MSPRF 0802915.}
\begin{document}

\begin{abstract}
A dinner table seats $k$ guests and holds $n$ discrete morsels of food.  Guests select morsels in turn until all are consumed.  Each guest has a ranking of the morsels according to how much he would enjoy eating them; these rankings are commonly known.

A {\em gallant knight} always prefers one food division over another if it provides strictly more enjoyable collections of food to one or more {\em other} players (without giving a less enjoyable collection to any other player) even if it makes his own collection less enjoyable.  A {\em boorish lout} always selects the morsel that gives him the most enjoyment on the current turn, regardless of future consumption by himself and others.

We show the way the food is divided when all guests are gallant knights is the same as when all guests are boorish louts but turn order is reversed.  This implies and generalizes a classical result of Kohler and Chandrasekaran (1971) about two players strategically maximizing their own enjoyments.  We also treat the case that the table contains a mixture of boorish louts and gallant knights.

Our main result can also be formulated in terms of games in which selections are made by groups.  In this formulation, the surprising fact is that a group can always find a selection that is simultaneously optimal for each member of the group.
\end{abstract}

\maketitle

Suppose that $k$ guests are seated at a table containing $n$ discrete morsels of food (sushi rolls, say).  Each player $i$ has a strict ranking $<_i$ of the $n$ morsels according to how enjoyable they are to eat.  Players take turns selecting and consuming a morsel from the table, according to a fixed turn order $P_1, P_2, \ldots, P_m$, where $m \leq n$, and $P_t \in \{1,2,\ldots,k\}$ is the player who selects a morsel on the $t$\,th turn.

Our dinner guests do not simply play to maximize their own enjoyment.  Instead, each guest conforms to one of two stereotypes: he is either a \emph{gallant knight} or a \emph{boorish lout}, as explained below.

A \emph{plate} is a subset of the morsels.  Each player's ranking of the morsels determines a corresponding partial ordering on plates of a fixed size according to \emph{pairwise comparison}: Given two plates $A$ and $A'$ of the same size, write $A \leq_i A'$ if there is a bijection $f : A \to A'$ such that $a \leq_i f(a)$ for all $a \in A$.

Given a sequence $\aa = (a_1, \ldots, a_m)$ of morsel selections representing the play of a game, let $A_i(\aa) = \{a_t \mid P_t = i\}$ be the plate eaten by player $i$.  For each player we define two partial orders on play sequences, her \emph{knight order} and \emph{lout order}.
We define $\aa < \bb$ in player $i$'s knight order if either
	\begin{enumerate}
	\item[(K1)] $A_j(\aa) \leq_j A_j(\bb)$ for all $j \neq i$ and $A_j(\aa) <_j A_j(\bb)$ for at least one $j \neq i$; or
	\item[(K2)] $A_j(\aa) = A_j(\bb)$ for all $j \neq i$ and $A_i(\aa) <_i A_i(\bb)$.
	\end{enumerate}
If neither (K1) nor (K2) holds, and neither holds with the roles of $\aa$ and $\bb$ reversed, and $\aa \neq \bb$,
then $\aa$ and $\bb$ are incomparable in player $i$'s knight order.  This could happen for various reasons:
	\begin{itemize}
	\item Some player $j \neq i$ receives incomparable plates $A_j(\aa)$ and $A_j(\bb)$; or
	\item $A_j(\aa) < A_j(\bb)$ and $A_{j'}(\aa) > A_{j'}(\bb)$ for some $j,j' \neq i$; or
	\item $A_j(\aa) = A_j(\bb)$ for all players $j$.
	\end{itemize}

Player $i$ is a \emph{gallant knight} if his (partial or total) preference ordering on play sequences extends his knight order.  Thus, a gallant knight must prefer $\bb$ to $\aa$ if all \emph{other} players receive more enjoyable plates in $\bb$ than in $\aa$.  He must also prefer $\bb$ to $\aa$ if all other players receive identical plates in $\aa$ and $\bb$ but he himself receives a more enjoyable plate in $\bb$ (This can only happen if $m < n$, i.e., if not all food is consumed during the game).  Our definition only requires a gallant knight to be ``gallant'' in a fairly weak sense: he prefers $\bb$ over $\aa$ whenever $\bb$ is {\em obviously} more enjoyable for the other players given their food preference rankings.

We define $\aa < \bb$ in player $i$'s lout order if the following holds:
\begin{enumerate}
\item[(L)] $a_s=b_s$ for $1\leq s\leq t-1$ and $a_t <_i b_t$, for some $t$ such that $P_t=i$.
\end{enumerate}
Player $i$ is a \emph{boorish lout} if his preference ordering on play sequences extends his lout order.  In other words, a boorish lout cares only about his own enjoyment and discounts his own future enjoyment so much that he always selects the morsel that maximizes his enjoyment on the current turn (regardless of the effect this may have on future turns).

\begin{theorem}
\label{t.main}
Under optimal play, the division of the food among gallant knights playing in turn sequence $P_1, \ldots, P_m$ is the same as the division of the food among boorish louts playing in the reversed turn sequence $P_m, \ldots ,P_1$.
\end{theorem}

By \emph{optimal play} we mean the outcome of a subgame perfect Nash equilibrium, which makes sense for either partially or totally ordered preferences.  (We will recall the definition below.)  By the \emph{division of the food} we mean the final partition of morsels among the players, without regard to the order in which they were chosen. We will prove a more general statement that also applies to mixed tables of knights and louts (and even to individuals who act as knights on some turns and louts on others), which we state as Theorem \ref{t.main.full} below.  Although boorish louts and gallant knights are very different, the theorem gives a sense in which they are ``dual'' to each other.

Nash equilibria in general are hard to compute.  In our setting, we have a sequential game of perfect information which has a succinct description (i.e., the length of the game description is only logarithmic in the size of the game tree).
For such games, the problem of determining whether a given outcome is an optimal play is $\textsc{Pspace}$-complete by a reduction to the quantified boolean formula problem \cite[Theorem 5.8]{AGS}.  The significance of Theorem \ref{t.main} and Theorem \ref{t.main.full} is that they allow us to compute the food division resulting from optimal play when some or all of the players are gallant knights -- namely, reverse the turn order and play as boorish louts.  The food division when all players are boorish louts is trivial to compute in linear time (in the size of the input, $nk$).  Theorem \ref{t.main} and Theorem \ref{t.main.full} are strong in the sense that our definition of gallant knights makes only fairly weak assumptions (i.e., we assume only that their preferences extend fairly weak partial orderings).  It is somewhat surprising that these weak assumptions are sufficient to determine the division of food under any optimal play.

\subsection*{Related work}
In a two player game where all morsels are ultimately consumed and both players are gallant knights, we can swap the roles of $<_1$ and $<_2$ to obtain the game in which two players take turns selecting morsels each with the aim of maximizing her own total enjoyment (but unlike boorish louts, they do not discount the future). Kohler and Chandrasekaran \cite{KC} described the optimal play for this latter game in 1971 (see also \cite[Ch.\ 2]{BT2} and \cite{LS} for a recent survey).

Theorem~\ref{t.main} thus generalizes the $2$-player theorem of \cite{KC} to a class of games with any number of players.
Curiously, we do not know of any efficiently computable equilibrium for the most obvious generalization of the Kohler-Chandrasekaran game, in which $k \geq 3$ players take turns selecting items and each tries to maximize his own total enjoyment. Brams and Straffin point out a number of pathologies that arise in this game \cite{BS}.

\subsection*{Enjoyments and utilities}
One concrete way to realize an extension of the pairwise comparison ordering on plates is to assign enjoyments $e_i(a)$ which are real numbers representing the value of morsel $a$ to player $i$.  We require $e_i(a) \neq e_i(b)$ for $a \neq b$, and we can say that one plate is more enjoyable than another for player $i$ if the sum of the enjoyments is larger.  Given a play sequence $\aa = (a_1, \ldots, a_m)$, we let $E_i(\aa) = \sum_{t \; : \; P_t=i} e_i(a_t)$ denote the {\em total enjoyment} that the $i$th player receives.
Each leaf $\aa$ of the game tree in Figure~\ref{fig:notunique} is labeled by the enjoyment vector $(E_1(\aa),E_2(\aa))$.

We also remark that a concrete way to realize a preference ordering on game outcomes for the $i$th player is to use a {\em utility function} that assigns a real number to each play sequence $\aa$.  An example of a gallant knight's utility function is
\[U_i(\aa) = \alpha E_i(\aa)+ \sum_{j \not = i} E_j(\aa).\]
for sufficiently small $\alpha>0$.  Since Theorems \ref{t.main} and \ref{t.main.full} hold for any extensions of the knight and lout orderings, in particular they apply when preferences are defined by utilities in this way.

\subsection*{Equilibria for games with partially ordered preferences}
\begin {figure}[!ht]
\centering
\includegraphics [width=4in]{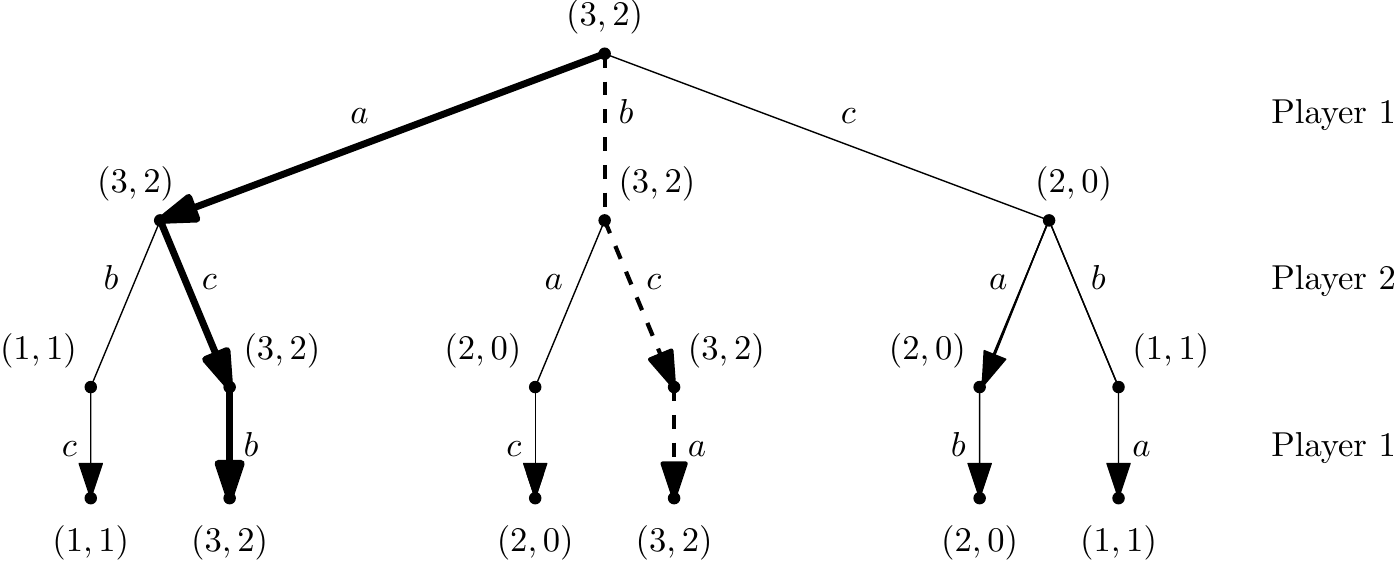}
\caption {\label{fig:notunique} Example of a game tree for $2$ gallant knights and $3$ morsels which has more than one optimal play.  The morsels $a$, $b$ and $c$ have enjoyment vectors $(1,0)$, $(2,1)$ and $(0,2)$ respectively.  Each edge is labeled with the morsel eaten by the moving player.  Each decision node contains one arrow pointing to a child.  The collection of all arrows represents an equilibrium pair of strategies, one for each player.  Each vertex is labeled by the enjoyment vector of the contingent outcome determined by the equilibrium.  The path of arrows from the root (left thickened path) is an optimal play.  The pair of strategies would still be an equilibrium if Player 1 chose the middle child instead of the left child for the first move; the resulting optimal play is the dotted path from the root.  Note that the division of food is the same in either case: Player 1 eats morsels $a$ and $b$ and Player 2 eats morsel $c$; only the order of eating changes.
}
\label{f.nonunique}
\end{figure}

The games we consider will involve finitely many states and outcomes and can be represented by decision trees such as Figure \ref{fig:notunique}.  A game outcome is a leaf of the tree.  Each player $i$ has a partial ordering $<_i$ on game outcomes.
A \emph{strategy} for player $i$ is a way of assigning a decision to each game position in which player $i$ has the move.  Each game position is represented by a node of the tree, and player $i$'s decision from a given position is represented by drawing an arrow from the corresponding node to one of its children.  A $k$-tuple of strategies (one strategy for each player) determines a directed path from the root of the tree to a leaf representing the outcome of the game.  The arrows also determine, for each node of the tree, a ``contingent outcome,'' which is the leaf reached by following the arrows starting from that node.

Given a $k$-tuple of strategies $s=(s_1,\ldots,s_k)$, for each node $x$ of the game tree let $O_s(x)$ be the set of contingent outcomes corresponding to the children of $x$.
 We say that $s$ is a \emph{subgame perfect Nash equilibrium} if for each node $x$ of the game tree, the contingent outcome of $x$ is a maximal element of $O_s(x)$ with respect to $<_i$, where $i$ is the player to move at $x$.
 A root-to-leaf trajectory is called an \emph{optimal play} if it arises from such an equilibrium.  It is always possible to construct such an equilibrium inductively (drawing decision arrows one layer at a time starting from the bottom of the tree).
If each player has a total preference ordering on game outcomes, then this equilibrium is unique.  If the orderings $<_i$ on game outcomes are partial instead of total,
then the equilibrium may not be unique.

The following proposition is immediate from the definitions, and relates equilibria for games with partially and totally ordered preferences on outcomes.

\begin{proposition} \label{p.orderequivalence} Suppose that $A$ and $B$ are $k$-tuples of partial orderings on game outcomes (one for each player), and each player's ordering in $B$ extends his ordering in $A$.  Then any equilibrium for $B$ is also an equilibrium for $A$. Conversely every equilibrium for $A$ is also an equilibrium for {\em some} $k$-tuple of total orderings that extends $A$.
\end{proposition}

Theorem~\ref{t.main} states that the knight and lout partial orderings contain enough information to determine how the food is divided under any optimal play (though, as Figure~\ref{f.nonunique} illustrates, they do not necessarily determine the {\em order} in which selections are made).  Another way to say this is that if we inductively choose a subgame perfect Nash equilibrium with respect to the partial ordering (i.e., at each node we draw an arrow to {\em any one} of the children corresponding to an outcome that is maximal with respect to the deciding player's partial ordering) then the resulting optimal play determines a partition of the morsels that does not depend on any of the choices we made along the way.

The following observation is also immediate.  It says that edges of the decision tree that correspond to unambiguously ``bad'' decisions can be removed without altering the set of optimal plays.

\begin{proposition} \label{p.removingedges} Suppose that $e_1, e_2, \ldots$ is a finite sequence of edges of the decision tree (corresponding to particular decisions made at nodes $x_1, x_2, \ldots$) and that there is no optimal play for any subgame beginning at any $x_j$ in which a player uses the edge $e_j$.  Then removing $e_1, e_2, \ldots$ (and the subtrees beneath them) from the decision tree does not alter the set of optimal plays.
\end{proposition}

\subsection*{Proof of the main result}

We will work in slightly greater generality by allowing {\em mixed nature} individuals.  A mixed nature individual can be understood as a gallant knight who knows she will be \emph{required} to act as a boorish lout on a certain pre-determined subset of her turns called her \emph{lout turns} (i.e., the edges of the decision tree corresponding to other moves she would have made at that step are just removed).  From the outset of the game, both the lout turns of each individual and the rankings $<_i$ on morsels are common knowledge to all players.

\begin{theorem}
  \label{t.main.full}
  In any game of mixed nature individuals, all optimal plays lead to the same division of food, which is the division obtained by pushing all knight turns to the end, reversing the order of the knight turns and then converting all knight turns into lout turns.
  \end{theorem}

For example, a game with $3$ players taking $6$ moves in the order
  	\[ \textbf{1(knight)},\;\; \textbf{2(knight)}, \;\;2(lout),\;\; 3(lout),\;\; \textbf{3(knight)},\;\; 1(lout) \]
would become
	\[ 2(lout),\;\; 3(lout),\;\; 1(lout),\;\; \textbf{3(lout)},\;\; \textbf{2(lout)},\;\; \textbf{1(lout)}. \]
The original knight turns, in \textbf{bold}, have been pushed to the end, reversed, and converted to lout turns.  This procedure determines a bijection between the turns of the original game and the turns of the resulting lout game.  Corresponding to the unique optimal play of the lout game under this bijection is a play of the original game, which the proof of Theorem \ref{t.main.full} will show is in fact optimal.

\begin{lemma} \label{switchlemma}
Let $G$ be a game of mixed nature individuals.  Assume that each player has preferences given by her knight ordering (rather than some extension of that ordering) but is forced to behave as a lout on her lout turns.  Let $t$ be the last turn on which any player acts as a gallant knight and suppose $t < m$.  Let $G'$ be the game obtained by reversing the turn order of $P_t$ and $P_{t+1}$. Then a food division can be obtained by an optimal play $\aa$ in $G$ if and only if it can be obtained by an optimal play $\aa'$ in $G'$.
\end{lemma}

\begin{proof}
Consider the pair of turns $t$ and $t+1$ as a single turn that results in a pair of morsels $(a,b)$ being distributed, the first to the boorish lout and the second to the gallant knight.  Since the lout's decision is deterministic, it is the knight's decision that determines the pair $(a,b)$.  If the knight goes second, then the lout will always take her favourite morsel from those remaining on turn $t$ (call it $a_1$) and the knight can choose any $b$ from the remaining available morsels.  If the knight goes first, he can also choose any such $b$ (in which case the lout chooses $a_1$).  However, the knight has another option in this case: he can choose $a_1$ (which would cause the lout to choose a second favourite option, $a_2$).

If the knight and lout are actually the {\em same player} then the options $(a_1,a_2)$ and $(a_2,a_1)$ are equivalent, so this does not actually correspond to an additional choice.  If the knight and lout correspond to different players, then the choice $(a_2,a_1)$ is always less desirable to the knight than $(a_1,a_2)$, since it results in the lout getting a less desirable option during these two turns (and since only louts remain, it does not affect the way food is distributed in subsequent turns).  Thus, having the knight move before the lout adds a choice for the knight (to each node of the decision tree corresponding to this combined move); but the added choice is one that the knight will never use under optimal play in any subgame, so (by Proposition \ref{p.removingedges}) adding or removing this choice has no effect on the way food is divided under equilibrium strategies.

To phrase this differently, we have shown that $\aa$ is optimal for $G$ if and only if $\aa'$ is optimal for $G'$, where $\aa'$ is the sequence obtained from $\aa$ by transposing $a_t$ and $a_{t+1}$, \emph{unless} $P_t = P_{t+1}$ and $a_t$ is $P_t$'s favourite remaining morsel, in which case $\aa'=\aa$.  We remark that it was important in this argument that all players had only the knight ordering on their preferences (rather than an extension of that ordering), since is this is the reason that we could treat the options $(a_1,a_2)$ and $(a_2,a_1)$ as equivalent.
\end{proof}

\begin{proof}[Proof of Theorem \ref{t.main.full}]
By Proposition~\ref{p.orderequivalence}, any equilibrium for players with extensions of the knight ordering is also an equilibrium for players with the knight ordering.
  By Lemma \ref{switchlemma}, for any optimal play $\aa$ of $G$, there is an optimal play $\aa'$ of the game $G'$ obtained by pushing the last knight turn to the end, such that $\aa$ and $\aa'$ lead to the same division of the food.  But when a gallant knight moves last, she behaves the same way she would behave if she were a boorish lout:  she takes her favourite morsel remaining.  Thus, $\aa'$ remains an optimal play if we convert her into a lout. Repeating this procedure with each knight in turn yields a game of only louts, which has a unique optimal play.
  \end{proof}

\begin{remark} The proof still works if the dinner table is a multiset, that is, it may contain multiple identical copies of some morsels (five identical tempura yam rolls, say). The only change to the argument is that, in the proof of Lemma \ref{switchlemma}, one must consider the case that the lout has more than one identical favourite item available at step $t$; but in this case, the lout will choose an item of that type regardless of the way $P_{t+1}$ and $P_t$ are ordered, so it does not matter which of the two goes first.

Note, however, that Lemma \ref{switchlemma} relies on the fact that when only lout turns remain, the play is deterministic.  For this reason, the proof does not apply to the more general setting in which each player has only a partial preference ordering on the morsels.  If some player regards two morsels as incomparable while another player has a preference between them, then on a lout turn, the former's choice may alter the outcome for the latter.  
\end{remark}

\subsection*{Another interpretation: games of group decision making}

\begin {figure}[!ht]
\centering
\includegraphics [width=5in]{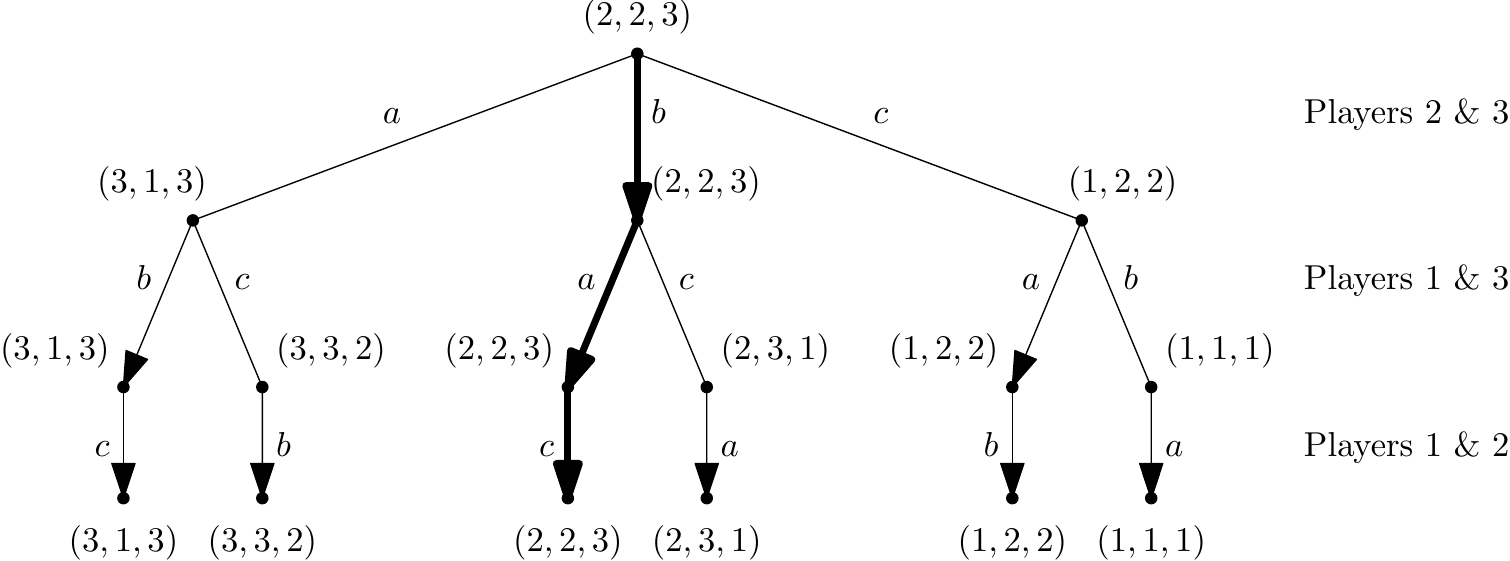}
\caption {\label{fig:exgame} An example group decision game tree for $3$ players and $3$ morsels.  The enjoyment vectors for the morsels $a$, $b$, and $c$ are $(0,1,2)$, $(1,2,1)$ and $(2,0,0)$ respectively.  Each vertex is labeled by the enjoyment vector of the contingent outcome determined by the equilibrium.  Each row is labeled with the players to agree upon the move at that row.  Each edge is labeled with the corresponding morsel they share.  Bold edges show the optimal play $b,a,c$.}
\label{f.group}
\end {figure}

Consider an alternate version of the dinner in which players take turns speaking.   While a player is busy speaking, the remaining $k-1$ players jointly select one of the morsels and share it among themselves.  Play continues until all food has been consumed (i.e., we assume the number $m$ of turns equals the number $n$ of morsels).  Different groups of $k-1$ players may have different methods of agreeing on a morsel to select as a group.  Let us assume only that their method is such that their choice is always {\em Pareto efficient} in the sense that they will prefer play sequence $\aa$ to $\bb$ if $\aa$ gives each member of the group an equal or more enjoyable plate (and at least one member a more enjoyable plate).

With these preferences, the game just described is equivalent (after reversing the enjoyment rankings for each player) to having the speaking player be a gallant knight who decides on a morsel to eat for himself.  To see the equivalence, note that if player $i$ partakes in a set of morsels $A_i$ in the group interpretation, that means he receives the complement $A_i^c$ in the knight interpretation.  The pairwise comparison ordering on sets $A_i$ corresponds to the comparison ordering on sets $A_i^c$ when the morsel ranking is reversed, and the Pareto ordering described above is exactly the knight ordering (here we use that $m=n$, so that axiom (K2) is irrelevant).

The tree for a group decision dinner of three players is shown in Figure \ref{f.group}, along with the ``enjoyment vector'' labeling shown as in Figure \ref{f.nonunique}.  Let us call an equilibrium (specified by a collection of strategies, one for each group) \emph{conflict-free} if no group can change a decision at a node unilaterally so that {\em any one} member of the group gets a more enjoyable plate (even if another member gets a less enjoyable plate).  Theorem~\ref{t.main} shows the following.

\begin{theorem}
For the group decision dinner, every subgame-perfect equilibrium is conflict-free, and all optimal plays result in the same division of the food.
  \label{t.group}
\end{theorem}

The statement that every subgame-perfect equilibrium is conflict-free means that as we construct an equilibrium inductively, at each node there exists at least one choice that is optimal for {\em all} of the players in the corresponding group.  This is immediate from Theorem~\ref{t.main}, which implies that if we are given any equilibrium leading to some given outcome $\aa$, then all of the options at a node $x$ leading to outcomes that are incomparable to $\aa$ (with respect to the group's Pareto ordering) must lead to the same division of food.  If one of these alternative outcomes was better for one player and worse for another player, then it would have to lead to a different division of a food, contradicting Theorem~\ref{t.main} (applied to the subgame beginning at $x$).

\end{document}